\documentclass[sn-mathphys,Numbered]{sn-jnl}

\usepackage{graphicx}%
\usepackage{multirow}%
\usepackage{amsmath,amssymb,amsfonts}%
\usepackage{amsthm}%
\usepackage{mathrsfs}%
\usepackage[title]{appendix}%
\usepackage{xcolor}%
\usepackage{textcomp}%
\usepackage{manyfoot}%
\usepackage{booktabs}%
\usepackage{algorithm}%
\usepackage{algorithmicx}%
\usepackage{algpseudocode}%
\usepackage{listings}%
\usepackage{bm}%

\newtheorem{theorem}{Theorem}
\newtheorem{proposition}[theorem]{Proposition}%
\newtheorem{lemma}[theorem]{Lemma}

\newcommand{\ue}{\mathrm{e}}
\newcommand{\ui}{\mathrm{i}}
\newcommand{\UI}{\mathrm{I}}
\newcommand{\ud}{\mathrm{d}}

\begin{document}


\title[Circulant graphs as an example of discrete quantum ergodicity]{Circulant graphs as an example of discrete quantum unique ergodicity}


\author*[1]{\fnm{Jon} \sur{Harrison}}\email{jon\_harrison@baylor.edu}

\author[1]{\fnm{Clare} \sur{Pruss}}\email{clare\_pruss1@baylor.edu}

\affil*[1]{\orgdiv{ Department of Mathematics}, \orgname{Baylor University}, \orgaddress{\street{1410 S. 4th Street}, \city{Waco}, \postcode{76706}, \state{TX}, \country{USA}}}


\abstract{

A discrete analog of quantum unique ergodicity was proved for Cayley graphs of quasirandom groups by Magee, Thomas and Zhao \cite{MTZ23}.  They show that for large graphs there exist real orthonormal basis of eigenfunctions of the adjacency matrix such that quantum probability measures of the eigenfunctions put approximately the correct proportion of their mass on subsets of the vertices that are not too small.  We investigate this property for Cayley graphs of cyclic groups (circulant graphs).   We observe that there exist sequences of orthonormal eigenfunction bases which are perfectly equidistributed.  However, for sequences of 4-regular circulant graphs of prime order, we show that there are no sequences of real orthonormal bases where all sequences of eigenfunctions equidistribute.  To obtain this result, we also prove that, for large 4-regular circulant graphs of prime order, the maximum multiplicity of the eigenvalues of the adjacency matrix is two.
}




\maketitle

\section{Introduction and Main Results}\label{sec intro}

Quantum Ergodicity (QE) was first identified as a property of the eigenfunctions of the Laplace-Beltrami operator on Riemannian manifolds \cite{S74, CdV85, Z87, GL93, ZZ96}.  When the geodesic flow is ergodic the probability density associated to eigenfunctions of the Laplacian tends to the uniform measure in the semiclassical limit for any subsequence of eigenfunctions of density one.\footnote{A subsequence that excludes a vanishingly small proportion of the eigenfunctions.}    Quantum ergodicity is a distinguishing feature of the quantum mechanics of a classically chaotic system, often referred to as quantum chaos.  More recently quantum systems with a stronger form of QE, Quantum Unique Ergodicity (QUE), have been identified.  For a system with QUE the sequence of probability measures associated to any subsequence of eigenfunctions tends to the uniform measure.  Consequently, a system with QUE cannot exhibit scarring, where probability measures associated to a subsequence of eigenfunctions tend to a delta measure localized on an unstable periodic orbit.  QUE was conjectured by Rudnick and Sarnak \cite{RS94} for compact hyperbolic manifolds and first proven by Lindenstrauss on arithmetic surfaces \cite{L06}.

QE has been investigated on quantum graphs in a number of settings \cite{KMW03,BKW04,BKS07,WB15}.  Quantum graphs are metric graphs, with a self adjoint operator, typically the Laplace or Schr\"odinger operator, acting on functions defined on a set of intervals associated with the edges, see \cite{IntroQuantumGraphs, GS06, SpecGeomGraphs}. 
In the case of discrete graphs investigation of the distribution of eigenfunctions was pioneered on regular  graphs \cite{S07,DP12,BL13,AM15,AS19}.  For example, in \cite{AM15} Anatharaman and Le Mason establish QE for families of large regular expanders (sparse graphs with high connectivity), where the expander condition plays the role of the ergodicity assumption in QE results for manifolds.    While the circulant graphs we consider are large regular graphs they are not expanders as the adjacency matrix has eigenvalues that are arbitrarily close to the largest eigenvalue in the large graph limit, see section \ref{sec spectra}.  
QUE was proved for random regular graph models in \cite{BKY17}.  Where, in the random graph model, equidistribution of the eigenfunctions holds for every sequence of eigenfunctions at the expense of holding for almost all sequences of graphs.  

For Cayley graphs QE was investigated in \cite{NSSZ23,SSZ22}.
In this article we consider the distribution of eigenfunctions of the adjacency matrix of large circulant graphs.  A circulant graph is a Cayley graph of the cyclic group $\mathbb{Z}_n$ where $n$ is the size of the graph, see Fig. \ref{fig: circulant graph 13} for an example of a circulant graph.  The undirected Cayley graph is defined by a choice of $m$ generators $a_1<a_2<\dots< a_m$ with $a_j \in \{1,\dots, n/2\}$ producing a $2m$-regular graph, see section \ref{sec circulant}.  We examine families of circulant graphs $C_n(\mathbf{a})$ in the large graph limit where the set of generators $\mathbf{a}=(a_1,\dots,a_m)$ is fixed.  
The distribution of the diameters of circulant graphs was investigated by Marklof and St\"ombergson \cite{MS13} and spectral properties of quantum circulant graphs were described in \cite{HS19} in the context of quantum chaology.  The number of periodic orbits on circulant graphs was evaluated in \cite{EHH??}.

\begin{figure}[htb]
\centering
\includegraphics[width=5cm]{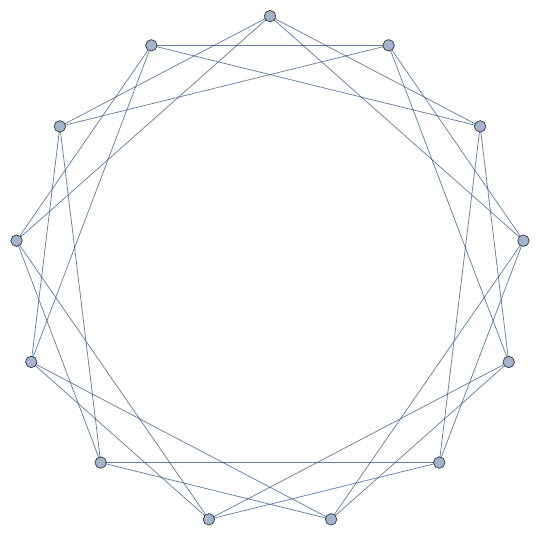}
\caption{The circulant graph $C_{13}(2, 3)$}\label{fig: circulant graph 13}
\end{figure}

Recently, Magee, Thomas and Zhao \cite{MTZ23} investigated a form of QUE for the discrete Cayley graphs of quasirandom groups.  An alternative proof avoiding representation theory has been obtained by Burq and Letrouit \cite{BL??}.
In \cite{MTZ23} the authors discuss a parallel of QUE for discrete graphs.  In analogy with the semiclassical limit of quantum graphs 
 \cite{GS06}, they consider a sequence of metric graphs $\{ G_j\}$ where $G_j$ has a vertex set $V_j$ and $|V_j|\to \infty$.  For each graph the operator is the $|V_j| \times |V_j|$ adjacency matrix $A(G_j)$ which records the connectivity of the graph; $[A(G)]_{uv}$ is one if the vertices $u$ and $v$ are adjacent in $G$ and zero otherwise.
The adjacency matrix acts on functions in $\ell^2(V_j)$
(which we also write $\ell^2(G_j)$).  
A normalized vector $\psi \in \ell^2(G)$ with $||\psi||_{\ell^2}=1$ defines a probability measure, 
\begin{equation}
    \mu_\psi = \sum_{v\in V} |\psi(v)|^2 \delta_v \ .
\end{equation}
where $\delta_v$ is the unit measure on the vertex $v$.  

Let $\{S_j\}$ be a sequence of vertex sets $S_j\subset V_j$ with $|S_j|/|V_j|\to p$, so we can think of vertex sets that contain approximately a proportion $p$ of the graph. Also fix a sequence $\{\varphi_j\}$ where $\varphi_j$ is a normalized eigenfunction of $A(G_j)$.  Following \cite{MTZ23} we say \emph{Discrete Quantum Unique Ergodicity} (DQUE) holds for the pair of sequences $\{\varphi_j\}$ and $\{S_j\}$ if 
\begin{equation}
    y_j=\mu_{\varphi_j}(S_j) \to p \ ,
\end{equation}
i.e. the measure associated to vertex sets by eigenfunctions in the sequence is in proportion to the size of the sets. This parallels QUE when this property holds for all sequences of vertex subsets of size $p$ and all sequences of eigenfunctions from a given sequence of bases.


Our first observation is that circulant graphs admit orthonormal bases where all sequences of eigenvectors and vertex subsets display DQUE.
\begin{proposition}\label{prop: Circulant DQUE}
For any sequence of circulant graphs $\{C_n(\mathbf{a}) \}_{n\in \mathcal{A}}$ 
there  exists a sequence 
$\{\mathcal{B}_n\}_{n\in \mathcal{A}}$ with $\mathcal{B}_n$ an orthonormal basis of eigenvectors of $A(C_n(\mathbf{a}))$, where 
\begin{equation}
    \lim_{n\to \infty } y_n = p \ .
\end{equation}
for any sequence of vertex subsets $S_n\subset \mathcal{V}(C_n(\mathbf{a}))$ such that $\lim_{n\to \infty} |S_n|/n =p$.
\end{proposition}  
We will see, in section \ref{sec DQUE}, that the statement follows directly from the formula for eigenvectors of circulant matrices which are perfectly equidistributed.  For comparison, in  \cite{BL??} the authors prove a general bound for elements of eigenvectors of the adjacency matrix of vertex transitive graphs, $|\psi (v) |\leq \sqrt{M}/n$ where $M$ is the maximum degeneracy of the eigenvalues, remark 3.1.
In section \ref{sec spectra}, we show that for sequences of $4$-regular circulant graphs of prime order the maximum multiplicity of eigenvalues of the adjacency matrix is two.

In \cite{MTZ23} the authors prove that the Cayley graphs of quasi-random groups admit real eigenfunction bases where all sequences of eigenfunctions and all sequences of vertex sets that are not too small exhibit DQUE.  They note that the proofs of their results simplify if one only looks for sequences of orthonormal bases, dropping the requirement that the bases be real.  In the case of circulant graphs we show that the sequence of bases cannot be made real without losing the DQUE property.  

\begin{theorem}\label{thm: circulant graph no real DQUE}
    For an increasing sequence of $4$-regular circulant graphs of prime order $\{C_n(a_1,a_2) \}_{n\in \mathcal{A}}$ there does not exist a sequence $\{\mathcal{B}_n\}_{n\in \mathcal{A}}$ with $\mathcal{B}_n$ a real orthonormal basis of eigenvectors of $A(C_n(a_1,a_2))$, where 
\begin{equation}
    \lim_{n\to \infty } y_n = p \ .
\end{equation}
for any sequence of vertex subsets $S_n\subset \mathcal{V}(C_n(a_1,a_2))$ such that $\lim_{n\to \infty} |S_n|/n =p$.
\end{theorem}




The article is laid out in the following manner.  In section \ref{sec circulant} we introduce circulant graphs.  In section \ref{sec spectra} we prove some required bounds on the multiplicity of the eigenvalues of $4$-regular circulant graphs, and in section \ref{sec DQUE} we investigate the DQUE property on circulant graphs.

\section{Circulant Graphs}\label{sec circulant}

A \emph{graph} $G$ is a set of \emph{vertices} $V(G) = \{0, 1, \dots, n-1\}$ and a set of \emph{edges} $E(G)$ consisting of unordered pairs of vertices, see, for example, Fig. \ref{fig: circulant graph 13}. Let $u,v\in V(G)$, if $\{u, v\} \in E(G)$ we say $u$ and $v$ are adjacent, written $u \sim v$.  The \emph{degree} of a vertex $v$ is the number of vertices adjacent to $v$ and a graph is $k$-\emph{regular} if the degree of every vertex is $k$.  

A \emph{simple graph} is a graph with no loops or multiple edges; so a vertex is not adjacent to itself and there is at most one edge between a pair of vertices.  Any graph can be made into a simple graph by introducing a dummy vertex of degree two on each loop and on each multiple edge. 
A graph is \emph{connected} if for every pair of vertices $u,v\in V(G)$ there is a sequence of vertices $u=p_1, p_2, \dots, p_{l+1}=v$ with $\{p_j,p_{j+1}\}\in E(G)$ for $j=1,\dots, l$.  

For the purposes of this paper, we will be interested in the class of simple circulant graphs which are Cayley graphs of cyclic groups.
A \emph{circulant graph} $C_n(a_1, a_2, \dots, a_m)$, where $1 \leq a_1 <\dots <a_m  \leq n/2$, is a graph with $V = \{0, 1, \dots, n-1\}$, where $i \sim j$ for $i, j \in V$ if $|i-j| = \pm a_k \,( \mathrm{mod }\, n)$ for some $k\in \{1 ,\dots , m\}$. We will also write $C_n(\mathbf{a})$ where $\mathbf{a} = (a_1, a_2, \dots, a_m)$.
%
Fig. \ref{fig: circulant graph 13} shows the $4$-regular circulant graph $C_{13}(1, 3)$. Note that the graph is connected; a circulant graph is connected if and only if $\mathrm{gcd}(n, a_1, a_2, \dots, a_m) = 1$.

For a simple graph we can define an adjacency matrix which records the connections between the vertices.
%
If $G$ is a simple graph with $n$ vertices, the $n\times n$ \emph{adjacency matrix} $A(G)$ is a matrix where $A_{ij} = 1$ if $i\sim j$ and $A_{ij} = 0$ otherwise.
%
For the circulant graph $C_{13}(2, 3)$, shown in Fig. \ref{fig: circulant graph 13}, the adjacency matrix is
\begin{equation}\label{eq:adjmatrix13}
    A( C_{13}(2, 3)) = 
    \left(
\begin{array}{ccccccccccccc}
        0 & 0 & 1 & 1 & 0 & 0 & 0 & 0 & 0 & 0 & 1 & 1 & 0\\
        0 & 0 & 0 & 1 & 1 & 0 & 0 & 0 & 0 & 0 & 0 & 1 & 1\\
        1 & 0 & 0 & 0 & 1 & 1 & 0 & 0 & 0 & 0 & 0 & 0 & 1\\
        1 & 1 & 0 & 0 & 0 & 1 & 1 & 0 & 0 & 0 & 0 & 0 & 0\\
        0 & 1 & 1 & 0 & 0 & 0 & 1 & 1 & 0 & 0 & 0 & 0 & 0\\
        0 & 0 & 1 & 1 & 0 & 0 & 0 & 1 & 1 & 0 & 0 & 0 & 0\\
        0 & 0 & 0 & 1 & 1 & 0 & 0 & 0 & 1 & 1 & 0 & 0 & 0\\
        0 & 0 & 0 & 0 & 1 & 1 & 0 & 0 & 0 & 1 & 1 & 0 & 0\\
        0 & 0 & 0 & 0 & 0 & 1 & 1 & 0 & 0 & 0 & 1 & 1 & 0\\
        0 & 0 & 0 & 0 & 0 & 0 & 1 & 1 & 0 & 0 & 0 & 1 & 1\\
        1 & 0 & 0 & 0 & 0 & 0 & 0 & 1 & 1 & 0 & 0 & 0 & 1\\
        1 & 1 & 0 & 0 & 0 & 0 & 0 & 0 & 1 & 1 & 0 & 0 & 0\\
        0 & 1 & 1 & 0 & 0 & 0 & 0 & 0 & 0 & 1 & 1 & 0 & 0
\end{array}
\right) . 
\end{equation}
Notice that the adjacency matrix is symmetric, and the adjacency matrix of a circulant graph is a circulant matrix; each row of the matrix is obtained by rotating the elements of the previous row one place to the right. 

The discrete Laplacian of a graph is $L=D-A$ where $A$ is the adjacency matrix and $D$ is a diagonal matrix of the vertex degrees \cite{SGT}.  Consequently, for a $2m$-regular circulant graph $L=2m \UI_n -A$ is also a circulant matrix.   An eigenvector of $A$ with eigenvalue $\lambda$ is also eigenvector of the Laplacian with eigenvalue $2m-\lambda$.  Consequently the results of sections \ref{sec spectra} and \ref{sec DQUE} on the multiplicity of the spectrum and equidistribution of eigenvectors apply equally to eigenvectors and eigenvalues of $L$.  

\section{Spectra of 4-Regular Circulant Graphs}
\label{sec spectra}

We will be interested in the distribution of elements of eigenvectors of the adjacency matrix.  
As the adjacency matrix is a circulant matrix the eigenvalues are  
\begin{equation}
    \lambda_j = \omega^{a_1j} + \omega^{-a_1j} + \omega^{a_2j} + \omega^{-a_2j} + ... + \omega^{a_mj} + \omega^{-a_dj},
\end{equation}
with $\omega=\ue^{\frac{2\pi \ui}{n} }$ for $j=0,\dots , n-1$.  So for a connected graph where $\mathrm{gcd}(n,\mathbf{a})=1$ we have a simple eigenvalue $\lambda_0=2m$. Then as $\lambda_j=\lambda_{n-j}$ all the eigenvalues except $\lambda_0$ (and $\lambda_{n/2}$ if $n$ is even) have a multiplicity of at least two.   

Notice also, that for $\mathbf{a}$ fixed,
\begin{equation}
    \lambda_j = 2 \sum_{l=1}^m \cos \left( \frac{2\pi a_l j}{n}\right)
\end{equation}
so there are eigenvalues $\lambda_j$ with $j\neq 0$ that are arbitrarily close to $\lambda_0=2m$ for $n$ large.  Hence, families of circulant graphs do not have the expander property.

To determine the maximum multiplicity of the eigenvalues of the adjacency matrix we consider $4$-regular circulant graphs $C_n(a_1,a_2)$ with $n$ prime.
\begin{proposition}\label{prop: spec multiplicity}
   For $n$ prime, if $n \nmid a_1^2+a_2^2$ then $A(C_n(a_1,a_2))$ has a simple eigenvalue $4$ and all the other eigenvalues have multiplicity two.  If $n \mid a_1^2+a_2^2$ then $\lambda_j = \lambda_k$ iff $n\mid j+k$ or $n \mid j^2 + k^2$. 
\end{proposition}
\begin{proof}
    Suppose $\lambda_j = \lambda_k$ for $j \neq k$. As $n$ is prime, $\{\omega^{a_1j}, \omega^{a_2j}, \omega^{-a_1j}, \omega^{-a_2j}\}$ and $\{\omega^{a_1k}, \omega^{a_2k}, \omega^{-a_1k}, \omega^{-a_2k}\}$ are two sets of primitive roots of unity. The cyclotomic field $\mathbb{Q}(\omega)$ is linearly independent over the rationals~\cite{Washington}. Hence, if $\lambda_j = \lambda_k$, we require $\{\omega^{a_1j}, \omega^{a_2j}, \omega^{-a_1j}, \omega^{-a_2j}\} = \{\omega^{a_1k}, \omega^{a_2k}, \omega^{-a_1k}, \omega^{-a_2k}\}$. Define two vectors,
    \begin{align}
    \mathbf{J} &= (\omega^{a_1j}, \omega^{a_2j}, \omega^{-a_1j}, \omega^{-a_2j}) \\
    \mathbf{K} &= (\omega^{a_1k}, \omega^{a_2k}, \omega^{-a_1k}, \omega^{-a_2k}) \ .    \end{align}
    If the two sets of primitive roots of unity are equal then $\mathbf{J} = \sigma(\mathbf{K})$   for some permutation $\sigma \in S_4$, where $\sigma(\mathbf{K})=(K_{\sigma(1)},\dots, K_{\sigma(4)})$.
     If $\sigma$ contains a $1$-cycle then $\mathbf{J} \neq \sigma(\mathbf{K})$ as then $\omega^{aj} = \omega^{ak}$ for some $a \in \{a_1, a_2, -a_1, -a_2\}$ and hence $\omega^{a(j-k)} = 1$ which is a contradiction as $\omega^a$ is a  primitive root of unity.    
There are nine permutations in $S_4$ that do not contain $1$-cycles, 
$(1 2) (3 4)$, $(1 3) (2 4)$, $(1 4) (2 3)$, $(1 2 3 4)$, $(1 2 4 3)$, $(1 3 2 4)$, $(1 3 4 2)$, $(1 4 2 3)$, and $(1 4 3 2)$.

 For $\sigma = (1 2) (3 4)$, if $\mathbf{J} = \sigma(\mathbf{K})$ then,
    \begin{align}
        \omega^{a_1j} &= \omega^{a_2k} \\
        \omega^{a_2j} &= \omega^{a_1k} \ .
    \end{align}
    Hence $\omega^{(a_1 + a_2)j} = \omega^{(a_1 + a_2)k}$, and $\omega^{(a_1 + a_2)(j-k)} = 1$. Therefore $j-k = 0\ (\mathrm{ mod }\ n)$ so $j = k$, which is a contradiction.

Similarly, for $\sigma = (1 4) (2 3)$, if $\mathbf{J} = \sigma(\mathbf{K})$, we have,
    \begin{align}
        \omega^{a_1j} &= \omega^{-a_2k}\\
        \omega^{-a_2j} &= \omega^{a_1k} \ .
    \end{align}
    This implies $\omega^{(a_1-a_2)(j-k)} = 1$, and $j - k = 0 \, (\mathrm{ mod }\, n)$ producing a contradiction.    
    
 For $\sigma = (1 3) (2 4)$, if $\mathbf{J} = \sigma(\mathbf{K})$, we have,
    \begin{align}
        \omega^{a_1j} &= \omega^{-a_1k} \\
        \omega^{a_2j} &= \omega^{-a_2k} \ .    
        \end{align}
    This implies $\omega^{a_i(j+k)} = 1$, so $j+k = 0\ (\mathrm{ mod }\ n)$ and wlog $k=n-j$, which we already noted produces degeneracy in the spectrum.  

  Now we will look at the $4$-cycles. Note that, if $\sigma(1)=2$ then $\sigma(3)=4$, as $\omega^{a_2j} = \omega^{-a_2k}$ implies $\omega^{-a_2j} = \omega^{a_2k}$. Similarly, if $\sigma(1)=3$, then $\sigma(3)=1$ and if $\sigma(1)=4$ then $\sigma(3)=2$.  This rules out all the $4$-cycles except $(1 2 3 4)$ and $(1 4 3 2)$. 

Consider $\sigma = (1 2 3 4)$. 
If $\mathbf{J} = \sigma(\mathbf{K})$, we have,
    \begin{align}
        \omega^{a_1j} &= \omega^{a_2k} \label{eq:1234 1}\\ 
        \omega^{a_2j} &= \omega^{-a_1k} \label{eq:1234 2} \\ 
        \omega^{-a_1j} &= \omega^{-a_2k} \label{eq:1234 3}\\ 
        \omega^{-a_2j} &= \omega^{a_1k} \ . \label{eq:1234 4}
    \end{align}
    Multiplying (\ref{eq:1234 2}) and (\ref{eq:1234 4}), we get  $\omega^{a_1k + a_2j} = \omega^{-(a_1k + a_2j)}$. Since $\omega^2$ is a primitive root of unity, $a_1k + a_2j = 0 \, (\mathrm{ mod }\ n)$.
Similarly, multiplying (\ref{eq:1234 1}) and (\ref{eq:1234 3}) gives us $\omega^{-a_1j+a_2k} = \omega^{-(-a_1j+a_2k)}$. So $a_1j = a_2k \, (\mathrm{ mod }\ n)$.  
Hence $0 = a_1^2k + a_1a_2j = (a_1^2+a_2^2)k \, (\mathrm{ mod }\ n)$. So $a_1^2 + a_2^2 = 0\, (\mathrm{ mod }\ n)$. Note that, equations (\ref{eq:1234 1}) through (\ref{eq:1234 4}) also imply $j^2 + k^2 = 0\, (\mathrm{mod}\ n)$. A similar argument produces the same conclusions if $\sigma = (1 4 3 2)$.    

    Suppose $a_1^2 + a_2^2 = 0\, (\mathrm{mod}\ n)$ and $j^2 + k^2 = 0\, (\mathrm{mod}\ n)$.  Then $\omega^{a_1j} = \omega^{(a_1^2j^2)^{1/2}} = \omega^{((-a_2^2)(-k^2))^{1/2}}=\omega^{\pm a_2k}  $. Similarly, $\omega^{a_2j}=\omega^{\pm a_1k}$. Hence $\{\omega^{a_1j}, \omega^{a_2j}, \omega^{-a_1j}, \omega^{-a_2j}\} = \{\omega^{a_1k}, \omega^{a_2k}, \omega^{-a_1k}, \omega^{-a_2k}\}$ and $\lambda_j  = \lambda_k$.
\end{proof}

For any particular choice of $(a_1, a_2)$, there are only finitely many primes $n$ such that $a_1^2 + a_2^2=0 \, (\mathrm{mod}\ n)$.  This implies we can always choose a sequence of graphs of increasing size with $(a_1, a_2)$ fixed for which the eigenvalues have a maximum multiplicity of two.

To illustrate proposition \ref{prop: spec multiplicity} for a graph where there are eigenvalues of the adjacency matrix with multiplicity greater than two, consider again $C_{13}(2, 3)$, shown in Fig. \ref{fig: circulant graph 13}, where $13|(2^2 + 3^2)$.  The adjacency matrix of this graph, see (\ref{eq:adjmatrix13}), 
has a simple eigenvalue $\lambda_0 = 4$ and three eigenvalues with multiplicity four,
\begin{align}
    \lambda_1 &= \lambda_5 = \lambda_8 = \lambda_{12} =2\sin \pi/26 +2\sin 5\pi /26 \ ,\\
    \lambda_2 &= \lambda_3 = \lambda_{10} = \lambda_{11} =-2\cos \pi/13 -2\sin 3\pi/26 \ ,\\
    \lambda_4 &= \lambda_6 = \lambda_7 = \lambda_9 
    =2\cos 2\pi/13 -2\cos 3\pi/13 \ .
\end{align}
To see this, note that $1 + 12 = 13$, so $\lambda_1 = \lambda_{12}$. Since $1^2 + 5^2 = 2\times 13$, we also have $\lambda_1 = \lambda_5$. Similarly, $1^2 + 8^2 = 5\times 13$, so $\lambda_1 = \lambda_8$. Note also that $8 + 5 = 13$, that $5^2 + 12^2 = 13\times 13$, and that $8^2 + 12^2 = 16\times 13$, hence the eigenvalues $\lambda_1$, $\lambda_5$, $\lambda_8$ and $\lambda_{12}$ are all equal. The equalities among the other eigenvalues are obtained similarly.


\section{Discrete Quantum Unique Ergodicity}\label{sec DQUE}


The standard formula for eigenvectors of an $n\times n$ circulant matrix is,
   \begin{equation}
        \mathbf{v}_j = \frac{1}{\sqrt{n}} (1, \omega^j, \omega^{2j}, \dots, \omega^{(n-1)j})
    \end{equation}
    with $j = 0, 1, ..., n-1$.  Hence $\mathcal{B}_j=\{\mathbf{v}_0,\dots, \mathbf{v}_{n-1}\}$ is an orthonormal basis of $A(C_n (\mathbf{a}))$.  Consequently, for $S_n \subset V_n$ we see,
    \begin{equation}
         y_n = \sum_{i \in S_n} |[\mathbf{v}_j]_i|^2 = \frac{|S_n|}{n},
    \end{equation}
which proves proposition \ref{prop: Circulant DQUE}.

We will now consider the DQUE property with the additional requirement that the orthonormal bases of eigenvectors of the adjacency matrix are real.
The eigenvalue $\lambda_0=2m$ has a real eigenvector $\mathbf{v}_0 = \frac{1}{\sqrt{n}}(1, \dots, 1)$.  For $a_m<n/2$ the other eigenvalues all have a multiplicity of at least two as $\lambda_j=\lambda_{n-j}$ 
for $j = 1, 2, \dots, n-1$. 
We can define a real orthonormal basis of eigenvectors of the adjacency matrix. For $j = 1, 2, \dots, (n-1)/2$ let,
\begin{align}
    \mathbf{c}_j &= \frac{1}{\sqrt{2}}(\mathbf{v}_j + \mathbf{v}_{n-j}) \ , \\
    \mathbf{s}_j &= \frac{1}{\sqrt{2}i}(\mathbf{v}_j - \mathbf{v}_{n-j}) \ .
\end{align}
Notice, 
\begin{align}
    \mathbf{c}_j &= \sqrt{\frac{2}{n}}\left(1, \cos\left(\frac{2\pi j}{n}\right), \cos\left(\frac{4\pi j}{n}\right), \dots , \cos\left(\frac{2(n-1)\pi j}{n}\right)\right) \ , \\
    \mathbf{s}_j &= \sqrt{\frac{2}{n}}\left(0, \sin\left(\frac{2\pi j}{n}\right), \sin\left(\frac{4\pi j}{n}\right), \dots, \sin\left(\frac{2(n-1)\pi j}{n}\right)\right) \ .
\end{align}
The set of eigenvectors $ \left\{\mathbf{v}_0, \mathbf{c}_1, \mathbf{s}_1, \dots, \mathbf{c}_{\frac{n-1}{2}}, \mathbf{s}_{\frac{n-1}{2}}\right\}$ is a real orthonormal basis for the eigenspace of $A(C_n(\mathbf{a}))$.

\begin{proposition}\label{lem: NDQUE}
Let $\{C_n(\mathbf{a}) \}_{n\in \mathcal{A}}$ be a sequence of $2d$-regular circulant graphs and 
consider the real orthonormal bases of eigenvectors of $A(C_n(\mathbf{a}))$, 
    \begin{equation}
        \mathcal{B}_n = \left\{\mathbf{v}_0, \mathbf{c}_1, \mathbf{s}_1, \dots, \mathbf{c}_{\frac{n-1}{2}}, \mathbf{s}_{\frac{n-1}{2}}\right\}.
    \end{equation}
    There exists vertex subsets $\{S_n\}$ and sequences of eigenvectors in $\mathcal{B}_n$ which do not exhibit  
    DQUE.
\end{proposition}

\begin{proof}
Let $0<p < \sqrt{6}/\pi \approx 0.780$.  For each $n\in \mathcal{A}$ consider the subset of vertices 
$S_n = \{n-m, n-m+1,\dots, n-1, 0, 1, 2, \dots, m\}$, where $m = \lfloor np/2\rfloor$, hence $\lim_{n \rightarrow \infty}|S_n|/n = p$.  If we use the sequence of basis vectors $\{ \mathbf{s}_{(n-1)/2} \}_{n\in \mathcal{A}}$ to define a sequence of probability measures on $\mathcal{V}$ then the measure of $S_n$ is,
    \begin{equation}
        y_n = \sum_{k \in S_n} \left|[\mathbf{s}_{(n-1)/2}]_k
        \right|^2 \\ 
        = \frac{4}{n} \sum_{k=1}^m \sin^2\left(\frac{\pi k}{n}\right) \ .  
        \end{equation}
    Since $|\sin x |< |x|$ for $x\neq 0$, we have
    \begin{equation}
        y_n < \frac{4\pi^2}{n^3} \sum_{k=0}^m k^2  \\
        \leq \frac{\pi^2 (np)(np+1)(np+2)}{6 n^3} \ .
    \end{equation}
    So $\lim_{n \rightarrow \infty}y_n \leq \pi^2p^3 / 6 < p$.  Hence there exist subsets of vertices where $\lim_{n \rightarrow \infty}y_n \neq \lim_{n \rightarrow \infty}|S_n|/n = p$.
\end{proof}
The proof of proposition \ref{lem: NDQUE} shows that for the bases $\{ \mathcal{B}_n \}_{n\in \mathcal{A}}$ there exists a sequence of eigenfunctions that does not equidistribute, even over subsets of vertices containing more than three quarters of each of the graphs.  

If eigenvalues of  $A(C_n(\mathbf{a}))$ have a maximum multiplicity of two 
we can define a general real orthonormal eigenvector basis.  Let $\bm{\theta}_n=(\theta_1,\dots, \theta_{(n-1)/2})$ with $0 \leq \theta_j \leq \pi/2$. Then let,
\begin{align}
    \mathbf{c}_j' &= \cos \theta_j \, \mathbf{c}_j + \sin \theta_j  \, \mathbf{s}_j \ , \\
    \mathbf{s}_j' &= -\sin \theta_j \, \mathbf{c}_j + \cos \theta_j  \,\mathbf{s}_j \ .
\end{align}
The set
\begin{equation}
    \mathcal{B}_n (\bm{\theta}_n)
    =
    \left\{\mathbf{v}_0, \mathbf{c}_1', \mathbf{s}_1', \dots, \mathbf{c}_{\frac{n-1}{2}}', \mathbf{s}_{\frac{n-1}{2}}' \right\}
\end{equation}
is a real orthonormal basis of eigenvectors of $A(C_n(\mathbf{a}))$. 


\begin{proposition}
  If the eigenvalues of $A(C_n(a_1, a_2))$ have maximum multiplicity two then every real orthonormal basis of $A(C_n(a_1, a_2))$ has the form 
    $\mathcal{B}_n (\bm{\theta}_n)$
    for some $\bm{\theta}_n \in [0, \pi/2]^{(n-1)/2}$.
\end{proposition}

\begin{proof}
    Consider the basis $\mathcal{B}_n = \{\mathbf{v}_0, \mathbf{c}_j, \mathbf{s}_j\}_{j=1}^{(n-1)/2}$. Let $V_j$ be the subspace of $\mathbb{R}^n$ spanned by $\{\mathbf{c}_j, \mathbf{s}_j\}$. A real orthonormal eigenvector basis of $A(C_n(\mathbf{a}))$ has the form $\{\mathbf{v}_0, \mathbf{u}_j, \mathbf{w}_j\}_{j=1}^{(n-1)/2}$ where $\{\mathbf{u}_j, \mathbf{w}_j\}$ span $V_j$. Hence,
    \begin{align}
    \mathbf{u}_j &= \alpha_j \mathbf{c}_j + \beta_j \mathbf{s}_j\\
    \mathbf{w}_j &= \gamma_j \mathbf{c}_j + \delta_j \mathbf{s}_j 
    \end{align}
    for some $\alpha_j, \beta_j, \gamma_j, \delta_j \in \mathbb{R}$. Since $|\mathbf{u}_j| = |\mathbf{w}_j| = 1$, we have $\alpha_j=\cos \theta_{j}$ and $\beta_j= \sin \theta_j $ for $\theta_j\in [0,\pi/2]$ and $\gamma_j = \cos \phi_j$ and $\delta_j = \sin \phi_j$ for $\phi_j\in [0,\pi]$ . Since $\mathbf{u}_j \cdot \mathbf{w}_j = 0$, we see $\cos(\theta_j - \phi_j)=0$. Hence $\gamma_j = -\sin \theta_j$ and $\delta_j = \cos \theta_j$.
\end{proof}

 Proposition \ref{lem: NDQUE} states that there exists sequences of vertex subsets and sequences of eigenvectors from the bases $\{\mathcal{B}_n (\mathbf{0}_n) \}_{n\in \mathcal{A}}$ that do not display DQUE.
In fact, for any sequence of bases $\{\mathcal{B}_n (\bm{\theta}_n) \}_{n\in \mathcal{A}}$ 
there will always be sequences of eigenvectors and vertex subsets 
which do not exhibit DQUE if the eigenvalues of $A(C_n(\mathbf{a}))$ have a maximum multiplicity of two.
\begin{proposition}\label{propn: no DQUE max mult 2}
Let $\{C_n(\mathbf{a}) \}_{n\in \mathcal{A}}$ be a sequence of $2d$-regular circulant graphs and 
consider any sequence $\{\mathcal{B}_n (\bm{\theta}_n) \}_{n\in \mathcal{A}}$
of real orthonormal bases of eigenvectors of $A(C_n(\mathbf{a}))$.  Then
    there exists vertex subsets $\{S_n\}$ and sequences of eigenvectors in $\mathcal{B}_n(\bm{\theta}_n) $  which do not exhibit  
    DQUE.
\end{proposition}

\begin{proof}
    Let $p < \sqrt{6}/\pi \approx 0.780$ and define a sequence of vertex subsets 
    \begin{equation}
        S_n = \left\{ 
        \left\lceil 
         - \frac{np}{2}-\frac{n\theta_{(n-1)/2}}{\pi}
       \right\rceil, 
       \left\lceil 
         - \frac{np}{2} - \frac{n\theta_{(n-1)/2}}{\pi}
       \right\rceil +1,
         \dots, \left\lfloor 
        \frac{np}{2} - \frac{n\theta_{(n-1)/2}}{\pi}
        \right\rfloor \right\}
    \end{equation}
    where the vertex $-j$ is interpreted as another label for the vertex $n-j$.
    Then $\mathrm{lim}_{n \rightarrow \infty}|S_n|/n = p$.  
        Note that,
    \begin{equation}
        \mathbf{s}_j' = \sqrt{\frac{2}{n}}\left( \sin\left(- \theta_j \right), \sin \left(\frac{2\pi j}{n} - \theta_j \right),   \dots,
     \sin\left(\frac{2(n-1)\pi j}{n} - \theta_j\right)\right) \ .
    \end{equation}    
    If we use the sequence of basis vectors $\left\{ \mathbf{s}_{(n-1)/2}' \right\}_{n\in \mathcal{A}}$ to define a sequence of probability measures on $\mathcal{V}$ then the measure of $S_n$ is,    
    \begin{align}
        y_n &= \sum_{k\in S_n} \frac{2}{n} \sin^2\left(\frac{\pi k}{n} + \theta_{(n-1)/2} \right)
        < \frac{4}{n}\int_0^{\frac{pn}{2}}  \left(\frac{\pi x}{n}\right)^2\ud x
        = \frac{\pi^2p^3}{6} \ .
    \end{align}
    Therefore $\mathrm{lim}_{n\rightarrow \infty}y_n \leq \pi^2p^3 / 6 < p$. Hence there exist subsets of vertices where $\lim_{n \rightarrow \infty}y_n \neq \lim_{n \rightarrow \infty}|S_n|/n = p$.
\end{proof}

Combining propositions \ref{prop: spec multiplicity} and \ref{propn: no DQUE max mult 2} produces theorem \ref{thm: circulant graph no real DQUE}.

\section{Concluding Remarks}\label{sec conclusions}

In this article we investigated the DQUE property on families of circulant graphs.  This example has a number of unusual features.  Circulant graphs are not expanders (sparse yet well connected graphs) a property that can fill the roll played by ergodicity in a proof of QE, so they might be expected not to exhibit eigenfunction equidistribution. 
 However, contrary to such an expectation, they posses eigenvector bases which are perfectly equidistributed over the graph.  So, for example, circulant families satisfy the DQUE property of \cite{MTZ23}, see proposition \ref{prop: Circulant DQUE}.  

 On the other hand, while recent results \cite{MTZ23, BL??} for eigenfunction equidistribution on the Cayley graphs of quasirandom groups also hold for real orthonormal eigenfunction bases, we see that this is not the case for circulant graph families.  We determine that there do not exist real orthonormal eigenfunction bases where all sequences of eigenfunctions have the DQUE property for families of $4$-regular circulant graphs of prime order.  We would anticipate that this will also hold for most families of circulant graphs, as such a result follows for any family where the maximum multiplicity of the eigenvalues of large graphs is two.  

So circulant graphs provide a quixotic model of eigenfunction equidistribution.  Possessing both perfectly equidistributed eigenfunctions in the absence of a well connected graph and discrete quantum unique ergodicity for orthonormal eigenfunction bases without discrete quantum unique ergodicity for real orthonormal bases.   

\bmhead{Acknowledgments}

JH would like to thank Cyril Letrouit and  Maricela Ramirez for helpful comments.
CP would like to thank the Baylor Honors Program where research for her honors thesis initiated this work.






\bibliography{jmh-bibliography}

\end{document}